\documentclass[12pt]{article}

\usepackage{amsmath}
\usepackage{amssymb}
\usepackage{amsfonts}
\usepackage{mathptmx} 
\usepackage{times}
\usepackage{graphicx}
\usepackage{color}

\newtheorem{theorem}{Theorem}
\newtheorem{lemma}{Lemma}

\def\<{\leqslant}           
\def\>{\geqslant}           

\def\x{\times}

\def\wh{\widehat}
\def\wt{\widetilde}

\def\C{\mathbb{C}}

\def\U{{\mathbb U}}

\def\K{{\mathbb K}}

\def\bfP{\mbox{${\bf P}$}}

\def\bE{\mathbf{E}}
\def\bD{\mathbf{D}}
\def\re{\mathrm{e}}
\def\rd{\mathrm{d}}
\def\rT{\mathrm{T}}
\def\bfK{\mbox{${\bf K }$}}
\def\bfL{\mbox{${\bf L }$}}
\def\bfM{\mbox{${\bf M }$}}
\def\bfS{\mbox{${\bf S }$}}
\def\bfT{\mbox{${\bf T }$}}

\def\bfA{\mbox{${\bf A}$}}
\def\bfH{\mbox{${\bf H}$}}

\def\cov{\mbox{${\bf cov}$}}

\def\L{{\mathbb L}}

\def\cZ{\mathcal{Z}}

\def\G{\mathbb{G}}
\def\E{\mathbb{E}}

\def\bfQ{\mbox{${\bf Q}$}}

\def\R{\mathbb{R}}

\def\S{\mathbb{S}}

\def\Z{\mathbb{Z}}

\def\phi{\varphi}

\def\esssup{\mathop{\rm ess\, sup}}

\def\Tr{\mathrm{Tr}}

\def\sn{{|\!|\!|}}

\input{epsf}

\begin{document}
\title{\LARGE\bf Anisotropy-based optimal filtering in linear discrete time
        invariant systems\footnote{This paper is a slightly edited version of the research report: I.Vladimirov, ``Anisotropy-based optimal filtering in linear discrete time
        invariant systems'', Centre for Applied Dynamical Systems, Mathematical Analysis and Probability, The University of Queensland, Brisbane, Australia,  CADSMAP Research Report 01--03, November 2001.  None of the original results have been removed,
nor have new results been added. The original work was supported by the Australian Research Council grant A~1002~7063 and was carried out while the author was with the Department of Mathematics, The University of Queensland, Brisbane, Australia.}
}
\author{Igor G. Vladimirov\thanks{UNSW Canberra, Campbell, ACT 2600, Canberra, Australia. E-mail: igor.g.vladimirov@gmail.com.
}}
\date{}

\maketitle
\begin{abstract}
This paper is concerned with a problem of robust filtering
for a finite-dimensional linear discrete time invariant system with
two output signals, one of  which is directly observed while the
other has to be estimated. The system is assumed to be driven by a
random disturbance produced from the Gaussian white noise sequence by an
unknown shaping filter. The worst-case performance of an estimator
is quantified by the maximum ratio of the root-mean-square (RMS) value of
the estimation error to that of the disturbance over stationary
Gaussian disturbances whose mean anisotropy is bounded from above
by a given parameter $a \>  0$. The mean anisotropy is a combined
entropy theoretic measure of temporal colouredness and spatial
``nonroundness'' of a signal. We construct an $a$-anisotropic estimator
which minimizes the worst-case error-to-noise RMS ratio. The estimator retains the
general structure of the Kalman filter, though with modified state-space
matrices. Computing the latter is reduced to solving a set of two
coupled algebraic Riccati equations and an equation involving the
determinant of a matrix. In two limiting cases, where $a = 0$ or $a
\to +\infty$, the $a$-anisotropic estimator leads to the
standard steady-state Kalman filter or the $H_{\infty}$-optimal
estimator, respectively.
\end{abstract}
\thispagestyle{empty}

\section{Introduction}

We consider a robust filtering problem for a finite-dimensional linear discrete time invariant (LDTI) system which generates two output signals. One of the signals is directly observed, while the other is unknown and has to be estimated by filtering the observation through a causal LDTI estimator. The underlying  system is driven by an external disturbance which is assumed to be a stationary Gaussian sequence. The latter is produced from a white noise sequence with zero mean and identity covariance matrix by an unknown LDTI shaping filter whose transfer function belongs to the Hardy space $H_2$.

For a given estimator and a given noise shaping filter, the influence of the disturbance on the estimation error can be quantified by the ratio of the root-mean-square (RMS) values of these random sequences. The aim of the estimator is to minimize this influence. Suppose the choice of a noise shaping filter is at the disposal of a hypothetical opponent whose aim is to maximize the  error-to-noise  RMS ratio. If the set of strategies of this opponent is the whole space $H_2$, then the worst damage to the estimator performance (in terms of the error-to-noise RMS ratio) which the opponent can achieve coincides with the $H_{\infty}$-norm of the \emph{error operator} $\Delta$ relating the estimation error to the disturbance. This effect is quantitatively the same as in a qualitatively different situation where the estimator performance is measured by the $\ell_2$-gain of the error operator with respect to nonrandom disturbances in the form of arbitrary square summable sequences. The latter setting is studied, for example, in \cite{NK_1991,PS_1999,S_1990,XDF_1991,XDW_1994}. On the other hand, if the opponent is allowed to produce only Gaussian white noise disturbances with scalar covariance matrices (that is, diagonal matrices with all equal diagonal entries), then the above mentioned RMS ratio reduces to the scaled $H_2$-norm $\|\Delta\|_2/\sqrt{m}$ of the error operator, where $m$ is the dimension of the disturbance.

Consider an intermediate situation where the opponent is restricted to produce disturbances whose \emph{mean anisotropy} is bounded from above by a given nonnegative parameter $a$ which quantifies the amount of uncertainty in the probability law of the noise.  The mean anisotropy is a combined entropy theoretic measure of temporal colouredness (that is, predictability) and spatial ``nonroundness'' of a stationary Gaussian sequence. The corresponding set of noise shaping filters is a cone in $H_2$. The worst-case performance of an estimator is then quantified by the \emph{$a$-anisotropic norm} of the error operator defined as the maximum error-to-noise RMS ratio with respect to the class of stationary Gaussian disturbances whose mean anisotropy does not exceed $a$.

The present paper is concerned with an \emph{$a$-anisotropic optimal filtering problem} of finding an estimator so as to minimize the $a$-anisotropic norm of the error operator. We derive equations for such estimator assuming the existence of a saddle point  in this stochastic minimax problem. The estimator retains the general structure of the Kalman filter, though with modified state-space matrices. Computing the latter is reduced to solving a set of two coupled discrete time algebraic Riccati equations (DAREs) and an equation involving the determinant of a matrix. In two limiting cases, where $a = 0$ and $a \to +\infty$, the $a$-anisotropic estimator becomes the standard steady-state Kalman filter \cite{AM_1979} and the $H_{\infty}$-optimal estimator, respectively. A numerical solution of this set of nonlinear equations can be implemented in the form of a homotopy algorithm, which is analogous to \cite{Mariton} and  employs the smooth parameter dependence of the stabilizing solutions of DAREs \cite{Ran1} and the vectorization of matrices \cite{Magnus}. The homotopy algorithm for the filtering problem, considered in the present paper, and its convergence will be discussed elsewhere.

The mean anisotropy of stationary Gaussian sequences and the anisotropic norm of LDTI systems were introduced in \cite{VKS_1995a}.    An account of their properties and applications to performance analysis of control systems can be found in \cite{DVKS_2001,VKS_1996a,VKS_1999}. The anisotropy-based optimization approach, which we follow in this paper, was proposed in \cite{SVK_1994,VKS_1995b} and then applied to optimal control  design for LDTI systems in \cite{VKS_1996b}. It pursues the aim of constructing controllers and estimators which would be more robust than the $H_2$-optimal ones  and less conservative than their $H_{\infty}$-counterparts. In this regard, the anisotropy-based approach is not dissimilar to  the mixed $H_2/H_{\infty}$-control \cite{DZJB_1994,HBM_1991,KRB_1996,ZGBD_1994} or to the minimum entropy $H_{\infty}$-control \cite{Mustafa} approaches where an internally stabilizing controller is sought to minimize the $H_2$-norm or, respectively, the entropy functional \cite{Arov}, of the closed-loop system under a given $H_{\infty}$-norm bound on the system. However, an important feature of our approach is that it deals with stochastic minimax  settings based on a single parameter-dependent norm which incorporates the standard $H_2$ and $H_{\infty}$ performance criteria as limiting cases.

The paper is organized as follows. In Section~\ref{preliminaries},
the mean anisotropy  and the anisotropic norm are defined and their
basic properties are outlined for convenience.
Section~\ref{problemstatement} formulates the $a$-anisotropic
optimal filtering problem. In Section~\ref{saddlepoint}, a
sufficient saddle point condition is provided for optimality of
an estimator in this problem. Section~\ref{finitedimensional}
specifies finite-dimensional estimators and noise shaping
filters among which the saddle point is being sought. In
Section~\ref{worstfilter}, equations are obtained for a worst-case
noise shaping filter against a given finite-dimensional estimator.
Section~\ref{weightedestimator} derives equations for the weighted $H_2$-optimal
estimator against a finite-dimensional shaping filter, and these
results are combined in Section~\ref{finalequations} which summarizes a set of algebraic equations for finding an $a$-anisotropic optimal estimator.

\section{Mean anisotropy of signals and $a$-anisotropic norm of systems}
\label{preliminaries}

In what follows, $V := (v_k)_{k \in \Z}$ denotes an $m$-dimensional Gaussian white noise sequence (of independent Gaussian random vectors in $\R^m$ with zero mean and the identity covariance matrix):
$$
    \bE  v_k = 0,
    \qquad
    \cov(v_j, v_k) = \delta_{jk} I_m,
    \qquad
    j, k \in \Z.
$$
Here, $\bE(\cdot)$ is the expectation, $\cov(\cdot, \cdot)$ is the covariance matrix, $\Z$ denotes the set of integers, $\delta_{jk}$ is the Kronecker delta, and $I_m$ is the identity  matrix of order $m$. Consider an $m$-dimensional stationary Gaussian sequence $W := (w_k)_{k \in \Z} := GV$ generated from $V$ by a causal LDTI shaping filter $G$ with an $\R^{m\x m}$-valued impulse response $g:= (g_k)_{k\in \Z_+}$ as the convolution of the latter with $V$:
\begin{equation}\label{wj}
        w_j := \sum_{k = 0}^{+\infty} g_k v_{j-k},
        \qquad
        j \in \Z.
\end{equation}
As a linear input-output operator (which maps $V$ to $W$), the filter $G$ is identified   with its $\C^{m\x m}$-valued transfer
function
$$
        G(z) := \sum_{k = 0}^{+\infty} z^k g_k,
        \qquad
        z \in \C,
$$
which is assumed to be in the Hardy space $H_2^{m\x  m}$ in order to ensure the
convergence of the series (\ref{wj}) in the mean square sense (and
hence,  with probability one in the Gaussian case). That is, the transfer function $G$ is
analytic in the open unit disc $\{z \in \C:\ |z| < 1\}$ of the
complex plane and has finite $H_2$-norm
\begin{equation}
\label{H2}
        \|G\|_2
        :=
                \sqrt{\frac{1}{2\pi}
                \int_{-\pi}^{\pi}
                \Tr \big(\wh{G}(\omega)
                \wh{G}(\omega)^*\big) \rd\omega}
        =
        \sqrt{
                \sum_{k = 0}^{+\infty}
                \Tr (g_k g_k^{\rT})}
        =
        \sqrt{\bE  (|w_0|^2)},
\end{equation}
where
\begin{equation}
\label{Ghat}
    \wh{G}(\omega)
    :=
    \sum_{k = 0}^{+\infty}
    \re^{ik\omega}g_k,
    \qquad
    \omega \in [-\pi,\pi],
\end{equation}
is the Fourier transform of the impulse response $g$. The quantity on the right-hand side of (\ref{H2}) is the RMS value of the sequence $W=GV$. The \emph{mean anisotropy} \cite{VKS_1995a}  of the sequence $W$ is computed as
\begin{equation}\label{meananiso}
        \overline{A}(G)
        =
        -\frac{1}{4\pi}
        \int_{-\pi}^{\pi}
        \ln\det
        \left(
        \frac{m
        \wh{G}(\omega)\wh{G}(\omega)^*}{\|G\|_2^2}
        \right)
        \rd\omega.
\end{equation}
This functional takes non-negative finite values for full rank shaping
filters $G$ (that is, satisfying $\det \wh{G}(\omega) \ne 0$ for
almost all $\omega \in [-\pi,\pi]$), and $\overline{A}(G):= + \infty$
otherwise. The mean anisotropy (\ref{meananiso}) is representable
as the sum of two nonnegative terms
\begin{equation}
\label{AA}
    \overline{A}(G)
    =
    \underbrace{-\frac{1}{2}
    \ln\det
    \left(\frac{m\,\cov(w_0)}{\bE  (|w_0|^2)}\right)}_{\rm nonroundness}
    +
    \underbrace{
    \frac{1}{2}
    \ln\det \left(\cov(w_0) \left(\cov\left(w_0 \mid (w_k)_{k < 0}\right)\right)^{-1}\right)}_{\rm colouredness},
\end{equation}
where $\cov(\cdot \mid \cdot)$ denotes the conditional covariance matrix (which is nonrandom in the Gaussian case being considered).
Here, the first term is zero only for scalar covariance matrices $\cov(w_0) = \lambda I_m$, with $\lambda > 0$, which correspond to isotropic Gaussian distributions in $\R^m$. The second term on the right-hand side of (\ref{AA}) is Shannon's mutual information \cite{CT_1991} between $w_0$ and the past history $(w_k)_{k < 0}$ of the Gaussian sequence $W$ and is closely related to the Szego-Kolmogorov formula \cite{R_1990}. This term vanishes  if and only if $w_0$ and $(w_k)_{k < 0}$ are statistically independent. Therefore, the mean anisotropy functional (\ref{meananiso}) is a combined entropy theoretic measure of \emph{spatial nonroundness} and \emph{temporal colouredness} (that is, predictability) of the stationary Gaussian sequence $W = GV$. In particular, $\overline{A}(G) = 0$ if and only if $W$ is a zero mean Gaussian white noise sequence with a \emph{scalar} covariance matrix.

Despite the relatively simple structure (\ref{meananiso}) and (\ref{AA}) of the mean anisotropy (due to which this functional can be calculated using state-space formulas \cite{VKS_1996a}), the right-hand side of (\ref{meananiso}) was obtained in \cite{VKS_1995a} as the following limit
\begin{equation}
\label{bD}
    \overline{A}(G) = \lim_{N\to +\infty} \frac{\bD_N(G)}{N}
\end{equation}
which constitutes the original definition of the mean anisotropy.
Here, $\bD_N(G)$ denotes the relative entropy \cite{CT_1991} of the probability distribution of the  $Nm$-dimensional normalised  random vector $\frac{W_N}{|W_N|}$ with respect to the uniform distribution over the unit sphere in $\R^{Nm}$, with $W_N:= (w_k)_{0\< k<N}$ denoting a fragment of the Gaussian sequence $W = GV$. It is the limit relation (\ref{bD}) that motivates the term ``mean anisotropy'' for the deviation from Gaussian white noise sequences with scalar covariance matrices. Indeed, the fragments of the latter sequences have \emph{isotropic} Gaussian distributions (which are invariant under the group of rotations), and the corresponding  normalised vectors are uniformly distributed over the unit spheres, in which case $\bD_N=0$ for any $N=1,2,3,\ldots$.

Now, let $F$ be a causal LDTI system with an $m$-dimensional input $W$ and an $r$-dimensional output $Z := (z_k)_{k \in \Z} =  FW$. Suppose its transfer function belongs to the Hardy space $H_{\infty}^{r\x  m}$, that is, the function is analytic in the open unit disc of the complex plane and has finite $H_{\infty}$-norm
$$
        \|F\|_{\infty}
        :=
        \sup_{|z| < 1} \sigma_{\max}(F(z))
        =
        \esssup_{\omega \in [-\pi,\pi]}
        \sigma_{\max}(\wh{F}(\omega)).
$$
Here, $\sigma_{\max}(\cdot)$ denotes the largest
singular value of a matrix, and $\wh{F}$ is the Fourier transform of the corresponding impulse response, in accordance with (\ref{Ghat}).  For a given $a \>  0$, the
\emph{$a$-anisotropic norm} \cite{VKS_1995a} of $F$ is defined by
$$
        \sn F\sn_a
        :=
        \sup_{W = GV:\, G \in \G_a}
        \sqrt{\frac{\bE  (|z_0|^2)}{\bE  (|w_0|^2)}} =
        \sup_{G \in \G_a}
        \frac{
        \|FG\|_2}{\|G\|_2},
$$
where
\begin{equation}\label{Ga}
        \G_a
        :=
        \big\{G \in H_2^{m\x  m}:\ \overline{A}(G) \<
        a\big\}
\end{equation}
is the set of shaping filters with the mean anisotropy (\ref{meananiso}) not exceeding the  threshold $a$. For any system $F \in H_{\infty}^{r\x  m}$, its $a$-anisotropic norm $\sn F\sn_a$ is a nondecreasing concave function of $a \>  0$ satisfying
\begin{equation}\label{range}
        \frac{\|F\|_2}{\sqrt{m}}
        =
        \sn F\sn_0 \<  \lim_{a \to +\infty} \sn F\sn_a =
        \|F\|_{\infty}.
\end{equation}

\section{Anisotropy-based optimal filtering problem}
\label{problemstatement}

Let $Z := (z_k)_{k \in \Z}$ be an unknown $r$-dimensional signal
which is to be estimated by using the measurements of a directly observed $p$-dimensional
signal $Y := (y_k)_{k \in \Z}$. Suppose these sequences are produced at
the output of  an LDTI system $F$ with an $n$-dimensional internal
state  $X := (x_k)_{k \in \Z}$ driven by an $m$-dimensional external
disturbance $W := (w_k)_{k \in \Z}$ according to the state-space equations
\begin{equation}\label{runxyz}
    \begin{bmatrix}
    x_{k+1}\\
    y_k\\
    z_k
    \end{bmatrix}
    =
    \begin{bmatrix}
    A & B\\
    C & D\\
    \Phi & \Psi
    \end{bmatrix}
    \begin{bmatrix}
    x_k\\
    w_k
    \end{bmatrix},
\end{equation}
where $A,B,C, D, \Phi, \Psi$ are appropriately dimensioned real
matrices. For what follows, we assume that $A$ is asymptotically
stable (that is, its spectral radius satisfies $\rho(A)<1$) and $D$ is of full row rank. The fact that the system $F$ has
the state-space representation (\ref{runxyz}) will be written as
\begin{equation}\label{FFF}
    \begin{bmatrix}
    Y \\
    Z
    \end{bmatrix}
    =
    \begin{bmatrix}
    F_1 W \\
    F_2 W
    \end{bmatrix}
    =
    F W,
    \qquad
    F =
    \begin{bmatrix}
    F_1 \\
    F_2
    \end{bmatrix}
    =
    \left[
    \begin{array}{c|c}
    A & B\\
    \hline
    C & D\\
    \Phi & \Psi
    \end{array}
    \right],
\end{equation}
where the subsystems
$$
    F_1
    :=
    \left[
    \begin{array}{c|c}
    A & B\\
    \hline
    C & D
    \end{array}
    \right],
    \qquad
    F_2
    :=
    \left[
    \begin{array}{c|c}
    A & B\\
    \hline
    \Phi & \Psi
    \end{array}
    \right]
$$
share the common state $X$ and map their common input $W$ to
the outputs $Y$ and $Z$, respectively. Let the observation $Y$ be
processed by an \emph{estimator} $E$ which is a causal LDTI system with an
$r$-dimensional output
\begin{equation}
\label{Zhat}
    \wh{Z} := (\wh{z}_k)_{k \in \Z} = EY .
\end{equation}
The corresponding sequence of estimation errors
\begin{equation}
\label{Ztilde}
    \wt{Z} :=
(\wt{z}_k)_{k \in \Z} := Z - \wh{Z} = \Delta(E)W
\end{equation}
is
the output of the system
\begin{equation}\label{Delta}
   \Delta(E) :=  F_2 - E F_1
\end{equation}
which we will refer to as the \emph{error operator}. In what follows, an estimator $E$ is said to be \emph{admissible} if the corresponding error operator satisfies $\Delta(E) \in H_{\infty}^{r\x  m}$. The set of admissible estimators for the system $F$ is denoted by $\E$. For a given $a \>  0$, we formulate the \emph{$a$-anisotropic optimal filtering problem} as the minimization of the $a$-anisotropic norm of the error operator (\ref{Delta}) over admissible estimators:
\begin{equation}\label{aniopt}
    {\rm minimize}
    \quad
    \sn \Delta(E)\sn_a
    :=
    \sup_{G \in \G_a}
    \frac{\|\Delta(E)G\|_2}{\|G\|_2}
    \quad
    {\rm over}\
    E \in \E,
\end{equation}
where $\G_a$ is the class of noise shaping filters given by (\ref{Ga}). This setting, which follows the anisotropy-based optimization approach \cite{SVK_1994,VKS_1995b}, is depicted in Fig.~\ref{fig1}.
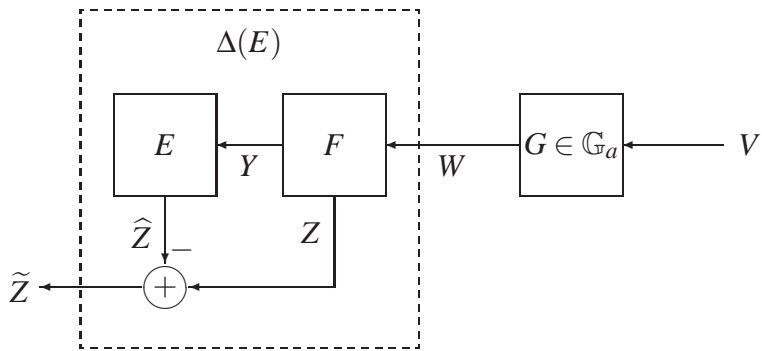
\begin{figure}[htbp]
\begin{center}
\unitlength=0.9mm
\linethickness{0.4pt}
\begin{picture}(110.00,100.00)

\put(75,62.5){\framebox(15,15)[cc]{$G \in \G_a$}}

\put(40,62.5){\framebox(15,15)[cc]{$F$}}
\put(35,85){\makebox(0,0)[cc]{$\Delta(E)$}}

\put(40,70){\vector(-1,0){10}}
\put(15,62.5){\framebox(15,15)[cc]{$E$}}
\put(47.5,62.5){\line(0,-1){13.5}}
\put(47.5,49){\vector(-1,0){21.5}}
\put(22.5,49){\makebox(0,0)[cc]{$+$}}
\put(25,54.5){\makebox(0,0)[cc]{$-$}}

\put(22.5,49){\circle{7}}
\put(22.5,62.5){\vector(0,-1){10}}
\put(19,49){\vector(-1,0){15}}
\put(1,49){\makebox(0,0)[cc]{$\wt{Z}$}}

\put(19,57){\makebox(0,0)[cc]{$\wh{Z}$}}
\put(44,57){\makebox(0,0)[cc]{$Z$}}

\put(35,67){\makebox(0,0)[cc]{$Y$}}
\put(65,67){\makebox(0,0)[cc]{$W$}}

\put(75,70){\vector(-1,0){20}}

\put(105,70){\vector(-1,0){15}}

\put(109,70){\makebox(0,0)[cc]{$V$}}

\put(10,40){\dashbox(50,50)[cc]{}}

\end{picture}
\vskip-3cm
\caption{The $a$-anisotropic optimal filtering problem.}
\label{fig1}
\end{center}
\end{figure}
If $a = 0$, then (\ref{aniopt})
coincides with the standard $H_2$-optimal filtering problem in view of the
left-most equality in (\ref{range}). On the other hand, the
limit on the right-hand side  of (\ref{range}) suggests
that,  for large values of the mean anisotropy level $a$, the problem (\ref{aniopt}) approaches the
$H_{\infty}$-optimal filtering problem. Note that, irrespective of whether
the minimum in the problem (\ref{aniopt}) is achievable, the quantity $\inf_{E \in \E} \sn
\Delta(E)\sn_a$ is a nondecreasing concave function of $a \>  0$ as
the lower envelope of such functions.

\section{Saddle-point condition of optimality}
\label{saddlepoint}

For a given mean anisotropy level $a \>  0$ of the disturbance $W$ and a given admissible estimator $E \in \E$, we denote by
\begin{equation}\label{worstnoise}
        \G_a^{\diamond}(E)
        :=
        \left\{G \in \G_a\ {\rm maximizing}\  \frac{\|\Delta(E) G\|_2}{\|G\|_2}\right\}
\end{equation}
the corresponding set of worst-case noise shaping filters. Furthermore, for any given noise shaping filter $G \in H_2^{m\x  m}$, let
\begin{equation}\label{bestestimator}
        \E^{\diamond}(G)
        :=
        \big\{
        E \in \E\
        {\rm minimizing}\
        \|\Delta(E) G\|_2\big\}
\end{equation}
denote the set of mean square optimal estimators which minimize the RMS value of the estimation errors in (\ref{Ztilde}) for the disturbance $W = GV$, thereby solving the weighted $H_2$-optimal filtering problem. The following lemma is similar to \cite[Lemma~1]{VKS_1996b}.

\begin{lemma}
\label{saddle}
Suppose $ E_* \in \E^{\diamond}(G_*)$ and $ G_* \in
\G_a^{\diamond}(E_*)$. Then the estimator $E_*$ is a solution of the $a$-anisotropic optimal filtering
problem (\ref{aniopt}).\hfill$\square$
\end{lemma}
\begin{proof}
In view of (\ref{worstnoise}) and (\ref{bestestimator}), the pair $(E_*,G_*)$, described in the lemma, is a
saddle point of the minimax problem (\ref{aniopt}). Hence, the relations
$$
    \sn \Delta(E) \sn_a
    \>
    \frac{\|\Delta(E) G_*\|_2}{\| G_*\|_2}
    \>
    \frac{\|\Delta(E_*) G_*\|_2}{\| G_*\|_2}
    =
    \sn \Delta(E_*) \sn_a
$$
hold for any admissible  estimator $E \in \E$, whereby the $a$-anisotropic norm of the error operator in (\ref{Delta}) can not be made smaller than that delivered by $E_*$.
\end{proof}

\section{Finite-dimensional estimators and noise shaping filters}
\label{finitedimensional}

In what follows, we will use an auxiliary set of matrices
\begin{equation}\label{setK}
    \K
    :=
    \big\{K \in \R^{n\x  p}:\ \rho(A-KC) < 1\big\}
\end{equation}
associated with the matrices $A$ and $C$ in (\ref{runxyz}).
Since $A$ is asymptotically stable, then $\K$ is an open subset of $\R^{n\x  p}$ which contains the zero matrix. For any $K \in \K$ and $M \in \R^{r\x  p}$, consider an estimator $E_{K,M}$ with an $n$-dimensional internal state $\wh{X} := (\wh{x}_k)_{k \in \Z}$ and the output $\wh{Z}$ in (\ref{Zhat}) governed by
\begin{equation}\label{EKMxz}
    \begin{bmatrix}
        \wh{x}_{k+1}\\
        \wh{z}_k
    \end{bmatrix}
    =
    \begin{bmatrix}
        A & K\\
        \Phi & M
    \end{bmatrix}
    \begin{bmatrix}
        \wh{x}_k\\
        y_k - C\wh{x}_k
    \end{bmatrix}
    =
    \begin{bmatrix}
        A-KC & K\\
        \Phi-MC & M
    \end{bmatrix}
    \begin{bmatrix}
        \wh{x}_k\\
        y_k
    \end{bmatrix},
\end{equation}
that is,
\begin{equation}\label{EKM}
        E_{K,M}
        =
        \left[
        \begin{array}{c|c}
        A - K C & K\\
    \hline
        \Phi-MC & M
        \end{array}
        \right].
\end{equation}
Note that the estimator $E_{K,M}$ has the structure of a steady-state Kalman
filter with gain matrices $K,M$. By introducing the sequence
\begin{equation}\label{diffX}
        \wt{X}
        :=
        (\wt{x}_k)_{k\in \Z}
        :=
        X - \wh{X}
\end{equation}
and using (\ref{runxyz}), it follows that
\begin{equation}
\label{yCx}
    y_k - C \wh{x}_k
    =
    C\wt{x}_k + D w_k .
\end{equation}
Substitution of (\ref{yCx}) into
(\ref{EKMxz}) leads to the following equations for $\wt{X}$ in (\ref{diffX}) and the estimation error
sequence $\wt{Z}$ in (\ref{Ztilde}):
\begin{equation}
\label{tildexz}
    \begin{bmatrix}
        \wt{x}_{k+1} \\
        \wt{z}_k
    \end{bmatrix}
    =
  \begin{bmatrix}
    A & B \\
    \Phi & \Psi
  \end{bmatrix}
  \begin{bmatrix}
    x_k \\
    w_k
  \end{bmatrix}
    -
    \begin{bmatrix}
        A & K\\
        \Phi & M
    \end{bmatrix}
    \begin{bmatrix}
        \wh{x}_k\\
        C\wt{x}_k + Dw_k
    \end{bmatrix}
    =
    \begin{bmatrix}
        A-KC & B-KD\\
        \Phi-MC & \Psi-MD
    \end{bmatrix}
    \begin{bmatrix}
        \wt{x}_k\\
        w_k
    \end{bmatrix}.
\end{equation}
Therefore, the error operator $\Delta(E)$ in (\ref{Delta}), which corresponds to the estimator $E:= E_{K,M}$ in
(\ref{EKM}), has the state-space representation
\begin{equation}\label{DKM}
    \Delta_{K,M}
    :=
        \Delta(E_{K,M})
        =
        \left[
        \begin{array}{c|c}
        A - K C & B-KD\\
    \hline
        \Phi -MC & \Psi-MD
        \end{array}
        \right]
\end{equation}
with the internal state $\wt{X}$ given by (\ref{diffX}). In view of (\ref{setK}), the condition $K \in \K$ ensures asymptotic stability of the system  $\Delta_{K,M}$ and hence, the admissibility of the estimator: $E_{K,M} \in \E$.
Now, with any matrix $K\in \K$, we associate the set
\begin{equation}\label{setLK}
    \L_K
     :=
    \big\{L \in \R^{m\x  n}:\
    \rho(A-KC + (B-KD)L) < 1\big\}
\end{equation}
which is an open subset of
$\R^{m\x  n}$ containing the zero matrix.  In what follows, $\S$ denotes the set of real positive definite
symmetric  matrices of order $m$. Consider
a noise shaping
filter $G_{K,S,L}$, which is parameterized by the matrices
\begin{equation}\label{KSL}
    K \in \K,
    \qquad
    S \in \S,
    \qquad
    L \in \L_K
\end{equation}
and produces a disturbance $W$ to the
system (\ref{runxyz}) as
\begin{equation}\label{shapew}
        w_k =  L \wt{x}_k + \sqrt{S} v_k,
\end{equation}
where $\sqrt{S}\in \S$ is the matrix square root of $S$,
and $\wt{X}$ is related by (\ref{diffX}) to the internal states $X$ and $\wh{X}$  of the underlying system (\ref{runxyz}) and the estimator (\ref{EKM}). This particular noise generation scenario is depicted in Fig.~\ref{fig2}.
\begin{figure}[htbp]
\begin{center}
\unitlength=0.9mm
\linethickness{0.4pt}
\begin{picture}(110.00,100.00)

\put(40,62.5){\framebox(15,15)[cc]{$F$}}
\put(40,70){\vector(-1,0){10}}
\put(15,62.5){\framebox(15,15)[cc]{$E_{K,M}$}}
\put(22.5,75.5){\line(0,1){15.5}}
\put(22.5,75.5){\makebox(0,0)[cc]{$\bullet$}}
\put(22.5,91){\vector(1,0){21.5}}
\put(47.5,62.5){\line(0,-1){13.5}}
\put(47.5,49){\vector(-1,0){21.5}}
\put(22.5,49){\makebox(0,0)[cc]{$+$}}
\put(25,54.5){\makebox(0,0)[cc]{$-$}}

\put(22.5,49){\circle{7}}
\put(22.5,62.5){\vector(0,-1){10}}
\put(19,49){\vector(-1,0){15}}
\put(1,49){\makebox(0,0)[cc]{$\wt{Z}$}}

\put(19,57){\makebox(0,0)[cc]{$\wh{Z}$}}
\put(44,57){\makebox(0,0)[cc]{$Z$}}

\put(19,84){\makebox(0,0)[cc]{$\wh{X}$}}
\put(44,84){\makebox(0,0)[cc]{$X$}}
\put(57,95){\makebox(0,0)[cc]{$\wt{X}$}}

\put(35,67){\makebox(0,0)[cc]{$Y$}}
\put(60,67){\makebox(0,0)[cc]{$W$}}

\put(47.5,75.5){\vector(0,1){12}}
\put(47.5,75.5){\makebox(0,0)[cc]{$\bullet$}}
\put(47.5,91){\makebox(0,0)[cc]{$+$}}
\put(42.3,93.5){\makebox(0,0)[cc]{$-$}}
\put(47.5,91){\circle{7}}
\put(51,91){\vector(1,0){12.5}}

\put(63.5,86){\framebox(10,10)[cc]{$L$}}
\put(68.5,70){\circle{7}}
\put(68.5,70){\makebox(0,0)[cc]{$+$}}
\put(68.5,86){\vector(0,-1){12.5}}
\put(65,70){\vector(-1,0){10}}

\put(82,70){\vector(-1,0){10}}
\put(82,65){\framebox(10,10)[cc]{$\sqrt{S}$}}
\put(107,70){\vector(-1,0){15}}

\put(111,70){\makebox(0,0)[cc]{$V$}}

\end{picture}
\vskip-3.5cm
\caption{The structure of the noise shaping filter $G_{K,S,L}$.}
\label{fig2}
\end{center}
\end{figure}
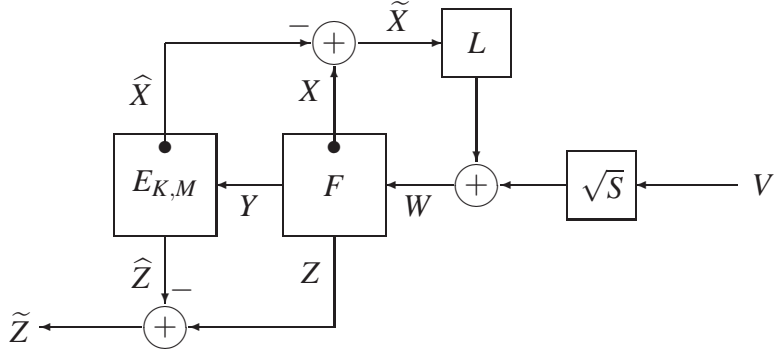
Substitution of (\ref{shapew}) into (\ref{tildexz}) leads to
\begin{align*}
        \wt{x}_{k+1}
         & =
        (A-KC)\wt{x}_k + (B-KD) w_k
          =
        (A-KC + (B-KD) L)\wt{x}_k + (B-KD) \sqrt{S}v_k,
\end{align*}
and hence, the noise shaping filter under consideration has the following state-space realization
\begin{equation}\label{GKSL}
        G_{K,S,L}
        =
        \left[
        \begin{array}{c|c}
        A-KC+(B-KD)L & (B-KD) \sqrt{S}\\
    \hline
    L & \sqrt{S}
        \end{array}
        \right]
\end{equation}
with the internal state $\wt{X}$ in (\ref{diffX}). The conditions (\ref{KSL}) on the
matrices $K, S, L$ imply that the filter $G_{K,S,L}$ and its inverse
\begin{equation}\label{invGKSL}
        G_{K,S,L}^{-1}
        =
        \left[
        \begin{array}{c|c}
        A-KC & B-KD\\
    \hline
    -S^{-1/2} L & S^{-1/2}
        \end{array}
        \right]
\end{equation}
are both asymptotically stable.
Therefore, $G_{K,S,L} \in H_2^{m\x  m}$ is a full rank noise shaping
filter which generates a stationary Gaussian disturbance $W=G_{K,S,L}V$ with a finite mean anisotropy $\overline{A}(G_{K,S,L}) < + \infty$.

\section{Worst-case noise shaping filter}
\label{worstfilter}

For any matrix pair $(K,M) \in \K\x  \R^{r\x  p}$, which specifies an admissible estimator $E_{K,M}$ in (\ref{EKM}), we associate with the
error operator $\Delta_{K,M}$ in (\ref{DKM}) a positive quantity
\begin{equation}\label{qmax}
    \theta_{K,M} := \|\Delta_{K,M}\|_{\infty}^{-2}
\end{equation}
which is a continuous function on the open
set $\K\x  \R^{r\x  p}$. For any $q \in [0,\theta_{K,M})$,
consider the following DARE with respect to a matrix $Q \in \R^{n\x  n}$:
\begin{align}
\label{dare1Q}
    Q
    & =
    (A-KC)^{\rT} Q(A-KC) + q (\Phi-MC)^{\rT} (\Phi-MC) +
    L^{\rT} S^{-1} L,\\
\label{dare1S}
    S
    & :=
    (I_m  -  (B-KD)^{\rT} Q (B-KD)-q (\Psi-MD)^{\rT} (\Psi-MD))^{-1},\\
\label{dare1L}
    L
    & :=
    S ((B-KD)^{\rT} Q (A-KC)+q (\Psi-MD)^{\rT} (\Phi-MC)).
\end{align}
A solution $Q$ of this equation will be called \emph{admissible} if it is symmetric and positive semi-definite, and $(S,L) \in \S\x  \L_K$. By the discrete bounded real lemma \cite{DeSouza} (see also \cite[Theorem~4.6.6 on p.~71]{SIG_1998}), for every $q \in [0,\theta_{K,M})$, the DARE (\ref{dare1Q})--(\ref{dare1L}) has a unique admissible solution. We will denote this solution and the associated matrices on the left-hand sides  of (\ref{dare1S}) and (\ref{dare1L}) by
\begin{equation}
\label{QSL}
    Q  =\bfQ(K,M,q),
    \qquad
    S  =  \bfS(K,M,q),
    \qquad
    L  =  \bfL(K,M,q).
\end{equation}
By a straightforward verification, if $q=0$ then,  for any $(K,M) \in
\K \x  \R^{r\x  p}$, these matrices reduce to
\begin{equation}\label{bfQSL0}
    \bfQ(K,M,0)   =   0,
    \qquad
    \bfS(K,M,0)   =   I_m,
    \qquad
    \bfL(K,M,0)   =   0.
\end{equation}
Moreover, by using the results of \cite{Ran1}, it can be shown
that the maps $\bfQ, \bfS, \bfL$ in (\ref{QSL}) are Frechet differentiable
on the set
\begin{equation}\label{setU}
    \U
    :=
    \big\{(K,M,q):\
    K \in \K,\
    M \in \R^{r\x  p},\
    0 \<  q < \theta_{K,M}
    \big\}
\end{equation}
which is defined in terms of (\ref{setK}), (\ref{DKM}) and (\ref{qmax}).

\begin{lemma}
\label{WCSF}
Suppose the matrices $S$ and $L$ are associated with the admissible solution of the DARE (\ref{dare1Q})--(\ref{dare1L}):
$$
    (K,M,q) \in \U,
    \qquad
    S = \bfS(K,M,q),
    \qquad
    L = \bfL(K,M,q).
$$
Then (\ref{GKSL}) describes a worst-case noise shaping filter against the
estimator (\ref{EKM}) in the sense of (\ref{worstnoise}). That is,
$$ G_{K,S,L} \in \G_a^{\diamond}(E_{K,M}),
$$
with
\begin{equation}\label{anorm}
        a = -\frac{1}{2} \ln\det \left(\frac{m S}{\Tr (L P L^{\rT}+S)}\right),
        \qquad
        \sn \Delta_{K,M} \sn_a
        =
        \sqrt{\frac{1}{q} \left(1-\frac{m}{\Tr (LPL^{\rT} +
        S)}\right)},
\end{equation}
where $P := \cov(\wt{x}_0)$ is the covariance matrix of the sequence (\ref{diffX}) under the
noise generation scenario $W = G_{K,S,L}V$.\hfill$\square$

\end{lemma}
\begin{proof}
The assertion of the lemma is a corollary from the results of \cite[Section 5]{VKS_1996a} (see also \cite[Lemmas~5 and 6, Theorem~4]{DVKS_2001}).
\end{proof}

In order to provide an additional insight into the structure of the filter $G_{K,S,L}$ described in
Lemma~\ref{WCSF}, we note that, in view of (\ref{invGKSL}), the system
\begin{equation}\label{auxsys}
        \Theta
        :=
        \begin{bmatrix}
        \sqrt{q} \Delta_{K,M}\\
        G_{K,S,L}^{-1}
        \end{bmatrix}
        =
        \left[
        \begin{array}{c|c}
        A-KC & B-KD\\
        \hline
        \sqrt{q} (\Phi-MC) & \sqrt{q}(\Psi-MD)\\
        -S^{-1/2} L & S^{-1/2}
        \end{array}
        \right]
\end{equation}
is inner, that is, its transfer function satisfies $\wh{\Theta}(\omega)^* \wh{\Theta}(\omega) = I_m$ for all $\omega \in [-\pi,\pi]$. Hence, under the noise generation scenario $W = G_{K,S,L}V$, the variance of the corresponding estimation errors in (\ref{Ztilde}) are related to that of the  disturbance $W$ by
$$
    q \bE (|\wt{z}_0|^2) + m = \bE  (|w_0|^2) = \Tr (L
    PL^{\rT} + S).
$$
Also note that the solution $Q$ of the Riccati equation (\ref{dare1Q})--(\ref{dare1L}) is  the observability gramian of the auxiliary system $\Theta$ in (\ref{auxsys}).

\section{Weighted $H_2$-optimal estimator}
\label{weightedestimator}

Using (\ref{setK}) and (\ref{setLK}), we will now introduce the set
$$ \L :=
\bigcup_{K \in \K} \L_K.
$$
For any matrix pair $(S,L) \in \S \x  \L$ in (\ref{shapew}),
consider the following DARE with respect to  a matrix $P \in \R^{n\x  n}$:
\begin{align}
\label{dare2P}
    P & =
    (A+BL) P (A+BL)^{\rT} + BS B^{\rT} - K T K^{\rT},\\
\label{dare2T}
    T
    & :=
    (C+DL)P(C+DL)^{\rT} + DS D^{\rT},\\
\label{dare2K}
    K & := \!\!
    ((A+BL)P(C+DL)^{\rT} + BS D^{\rT})T^{-1}.
\end{align}
A solution $P$ of this DARE will be called \emph{admissible} if
it is symmetric and positive semi-definite, and $\L_K \ni L$. Such a
solution, when it exists, is unique  \cite{Lancaster}. With the
admissible solution $P$, we associate the matrix
\begin{equation}\label{matM}
    M := ((\Phi+\Psi L) P (C+DL)^{\rT} + \Psi S
    D^{\rT}) T^{-1}.
\end{equation}
For what follows, the admissible solution $P$ and the associated matrices on the left-hand sides  of (\ref{dare2T})--(\ref{matM}) are denoted by
$$
    P =  \bfP(S,L),
    \qquad
    T  =  \bfT(S,L),
    \qquad
    K  =  \bfK(S,L),
    \qquad
    M  =  \bfM(S,L).
$$

\begin{lemma}
\label{bestEKM}
Suppose the matrices $K$ and $M$ are associated with the admissible solution of the DARE (\ref{dare2P})--(\ref{dare2K}) and (\ref{matM}):
$$
    K = \bfK(S,L),
    \qquad
    M = \bfM(S,L).
$$
Also, let  $K \in
\K$. Then (\ref{EKM}) describes a mean square optimal estimator against the
noise shaping filter (\ref{GKSL}) in the sense of (\ref{bestestimator}),
that is,
\begin{equation}
\label{EKMbest}
    E_{K,M} \in \E^{\diamond}(G_{K,S,L}).
\end{equation}
\hfill $\square$
\end{lemma}

\begin{proof}
This lemma will be proved by showing that, under its
assumptions, the output $\wh{Z}$ of the estimator $E_{K,M}$
reproduces the sequence $\cZ := (\zeta_k)_{k \in \Z}$ (that is,
$\wh{Z} = \cZ$) of conditional expectations
\begin{equation}\label{meanz}
    \zeta_k = \bE  (z_k \mid Y_k).
\end{equation}
The latter are computed according to the noise generation scenario $W = G_{K,S,L} V$ described by (\ref{runxyz}), (\ref{EKMxz}), (\ref{diffX}) and (\ref{shapew}) (see also Fig.~\ref{fig2}) and are, therefore, mean square optimal  in the sense of minimizing the RMS value $\sqrt{\bE(|z_k-\zeta_k|^2)}$ of the corresponding estimation errors in this scenario. Here, for any $k \in \Z$, we denote by $Y_k$ the $\sigma$-algebra of events generated by the past history $ (y_j)_{j \<  k} $ of the observation signal $Y$ available at the $k$th moment of time (that is, $(Y_k)_{k \in \Z}$ is the natural filtration of $Y$). Now, by introducing a sequence $\Xi := (\xi_k)_{k \in \Z}$ of the system state predictors
\begin{equation}\label{meanx}
        \xi_k := \bE  (x_k \mid Y_{k-1}),
\end{equation}
it follows that $\Xi$ and $\cZ$ satisfy the standard Kalman filtering equations (see, for example, \cite{AM_1979,LS_1977}):
\begin{equation}\label{generalfilter}
    \begin{bmatrix}
        \xi_{k+1}\\
        \zeta_k
    \end{bmatrix}
     =
    \bE
    \left(
    \left.
    \begin{bmatrix}
        x_{k+1}\\
        z_k
    \end{bmatrix}\,
    \right|\,
    Y_{k-1}
    \right)
    +
    \begin{bmatrix}
        K_*\\
        M_*
    \end{bmatrix}
    (y_k - \bE (y_k \mid Y_{k-1})),
\end{equation}
where the matrices $K_* \in \R^{n\x  p}$ and $M_* \in
\R^{r\x  p}$ are given by
\begin{align}
\label{P*}
    P_* & :=  \cov(x_{k+1} \mid Y_k)
     =
    \cov(x_{k+1} \mid Y_{k-1})
    -
    K_* T_* K_*^{\rT},\\
\label{T*}
    T_*
    & :=
    \cov(y_k \mid Y_{k-1}),\\
\label{K*}
    K_*
    & :=
    \cov(x_{k+1}, y_k \mid Y_{k-1})\,
    T_*^{-1},\\
\label{M*}
    M_*
    & :=
    \cov(z_k, y_k \mid Y_{k-1})\,
    T_*^{-1}.
\end{align}
Here, the conditional covariance matrices are nonrandom  and time invariant
since all the random sequences being considered are jointly Gaussian and stationary. The second equality in (\ref{P*}) follows from the Lemma on
Normal Correlation \cite{LS_1977}.  The equations (\ref{runxyz}),
(\ref{EKMxz}) and (\ref{shapew}) imply the inclusion
\begin{equation}\label{YV}
    Y_k \subset V_k
\end{equation}
for any $k \in \Z$,
where $V_k$ denotes the $\sigma$-algebra of events generated by the
past history $(v_j)_{j \<  k}$ of the Gaussian white noise sequence $V$
at the $k$th moment of time.  For every $k \in \Z$, the state prediction
error
$$
    \eta_k := x_k - \xi_k,
$$
associated with (\ref{meanx}), is a
$V_{k-1}$-measurable random vector, independent of the
$\sigma$-algebra $Y_{k-1}$.  Hence, by using (\ref{YV}) and (\ref{P*}), it follows that
\begin{equation}\label{Ecovetav}
    \bE
   \left(
    \left.
    \begin{bmatrix}
    \eta_k\\
    v_k
    \end{bmatrix}\,
    \right|\,
    Y_{k-1}
    \right)
   = 0,
    \qquad
    \cov
    \left(
    \left.
    \begin{bmatrix}
    \eta_k\\
    v_k
    \end{bmatrix}\,
    \right|\,
    Y_{k-1}
    \right)
    =
    \begin{bmatrix}
    P_* & 0 \\
    0 & I_m
    \end{bmatrix}.
\end{equation}
Note that, in view of (\ref{EKMxz}), the random vector $\wh{x}_k$ is
$Y_{k-1}$-measurable, which, in combination with  (\ref{shapew}) and the first of
the equalities (\ref{Ecovetav}), implies
that
$$
    \bE (w_k \mid Y_{k-1}) = L (\xi_k - \wh{x}_k).
$$
From the latter relationship and
from  (\ref{runxyz}), it follows that
\begin{align}
\label{Exyz}
    \bE
    \left(
    \left.
    \begin{bmatrix}
    x_{k+1}\\
    y_k\\
    z_k
    \end{bmatrix}\,
    \right|\,
    Y_{k-1}
    \right)
    & =
    \begin{bmatrix}
    A & B\\
    C & D\\
    \Phi & \Psi
    \end{bmatrix}
    \begin{bmatrix}
    \xi_k\\
    L (\xi_k-\wh{x}_k)
    \end{bmatrix},\\
\nonumber    \\
\label{covxyz}
    \cov
    \left(
    \left.
    \begin{bmatrix}
    x_{k+1}\\
    y_k\\
    z_k
    \end{bmatrix}\,
    \right|\,
    Y_{k-1}
    \right)
    & =
    \begin{bmatrix}
    A+BL & B\\
    C+DL & D\\
    \Phi+\Psi L & \Psi
    \end{bmatrix}
    \begin{bmatrix}
    P_* & 0 \\
    0 & S
    \end{bmatrix}
    \begin{bmatrix}
    A+BL & B\\
    C+DL & D\\
    \Phi+\Psi L & \Psi
    \end{bmatrix}^{\rT}.
\end{align}
Now, a combination of (\ref{Exyz}) with (\ref{generalfilter}) leads to
\begin{align*}
    \begin{bmatrix}
        \xi_{k+1}\\
        \zeta_k
    \end{bmatrix}
    & =
    \begin{bmatrix}
    A+BL & -BL\\
    \Phi+\Psi L & -\Psi L
    \end{bmatrix}
    \begin{bmatrix}
        \xi_k\\
        \wh{x}_k
    \end{bmatrix}
    +
    \begin{bmatrix}
    K_*\\
    M_*
    \end{bmatrix}
    (y_k - (C+DL)\xi_k + DL
    \wh{x}_k)
\end{align*}
which, along  with (\ref{EKMxz}),  implies that $\cZ = E_* Y$,
where the estimator $E_*$ has the state-space realization
\begin{equation}\label{E*}
    E_*
    =
    \left[
    \begin{array}{cc|l}
    A-K_*C +  (B-K_*D)L & -(B-K_* D)L & K_*\\
    0 & A-KC & K\\
    \hline
    \Phi- M_*C + (\Psi-M_* D)L &  - (\Psi-M_* D) L & M_*
    \end{array}
    \right]
\end{equation}
with a $2n$-dimensional internal state $(\Xi, \wh{X})$. By substituting
the covariance relations (\ref{covxyz}) into (\ref{P*})--(\ref{K*}), it follows
that the matrix $P_*$ with necessity satisfies the DARE
(\ref{dare2P})--(\ref{dare2K}). Moreover, the estimator
$E_*$  produces the sequences (\ref{meanz}) and (\ref{meanx})  if and only if
$P_*$ is an admissible solution of the Riccati equation. Hence, in view of
the above mentioned uniqueness of such solution, $P_* = \bfP(S,L),\
K_* = \bfK(S,L),\ M_* = \bfM(S,L)$ which, under the assumptions of the
lemma, implies that $K_* = K,\ M_* = M$. In order to complete the proof of the lemma, we will
need the following technical result.

\begin{lemma}
\label{simple}
Suppose $\alpha_{11}, \alpha_{12}, \alpha_{22} \in
\R^{n\x  n}$ are three matrices such that $\alpha_{11}$ and $\alpha_{22}$ are asymptotically
stable, and
$$
    \alpha_{11} + \alpha_{12} = \alpha_{22}.
$$
Then
for any $\beta \in \R^{n\x  p}$,\ $\gamma_1, \gamma_2 \in
\R^{r\x  n}$ and $\delta \in \R^{r\x  p}$, the following state-space
realizations determine the same input-output
operator:
$$
    \left[
    \begin{array}{cc|c}
    \alpha_{11} & \alpha_{12} & \beta\\
    0 &\alpha_{22} & \beta\\
    \hline
    \gamma_1 & \gamma_2 & \delta
    \end{array}
    \right]
    =
    \left[
    \begin{array}{c|c}
    \alpha_{22} & \beta\\
    \hline
    \gamma_1 + \gamma_2 & \delta
    \end{array}
    \right].
$$
\hfill $\square$
\end{lemma}
Now, by recalling (\ref{EKM}) and applying Lemma~\ref{simple} to
(\ref{E*}), it follows that $\Xi = \wh{X}$, $\cZ = \wh{Z}$
and $E_* = E_{K,M}$.  Therefore, the estimator, described in Lemma~\ref{bestEKM}, indeed satisfies (\ref{EKMbest}), which completes the proof.
\end{proof}

\section{Equations for $a$-anisotropic  optimal  estimator}
\label{finalequations}

The following theorem combines the results of Sections~\ref{saddlepoint}--\ref{weightedestimator} and provides a set of equations for finding an optimal estimator in the $a$-anisotropic filtering problem.
\begin{theorem}
\label{main}
Suppose the matrices $K \in \K$,\ $M \in \R^{r\x  p}$,\ $S \in \S$
and $L \in \L_K$
satisfy the equations
\begin{equation}\label{KMSL}
        K  =  \bfK(S,L),\qquad
        M  =  \bfM(S,L),\qquad
        S  =  \bfS(K,M,q),\qquad
        L  =  \bfL(K,M,q),
\end{equation}
where $q \in [0, \theta_{K,M})$. Then the estimator
(\ref{EKM}) is a solution of the $a$-anisotropic optimal
filtering problem (\ref{aniopt}), with the mean anisotropy level $a$ and the $a$-anisotropic norm $\sn \Delta_{K,M}\sn_a$ given by (\ref{anorm}),
where $P = \bfP(S,L)$. Here, the maps $\bfP$, $\bfK$, $\bfM$ and
$\bfS$, $\bfL$ are associated with the DAREs (\ref{dare2P})--(\ref{dare2K})
and (\ref{dare1Q})--(\ref{dare1L}), respectively, and the function $\theta_{K,M}$
is defined by (\ref{qmax}).\hfill $\square$
\end{theorem}
\begin{proof}
Under the assumptions of the theorem,
Lemmas~\ref{WCSF} and~\ref{bestEKM} imply that the estimator $E_{K,M}$
and the shaping filter $G_{K,S,L}$ satisfy the conditions of
Lemma~\ref{saddle}, whence the assertion of the theorem follows.
\end{proof}

The equations (\ref{KMSL}) can be written as
$$
    (K,M) = \bfH(K,M,q),
$$
where $\bfH$ is a Frechet differentiable map which is defined on the set $\U$ in (\ref{setU}) and is expressed in terms of $\bfK, \bfM, \bfS, \bfL$. Similarly, the first equality in (\ref{anorm}) determines a  function $\bfA: \U \to \R_+$ in terms of which the equations of Theorem~\ref{main} for the $a$-anisotropic optimal estimator take the form
$$
        (K,M) = \bfH(K,M,q),
        \qquad
        \bfA(K,M,q) = a,
        \qquad
        (K,M,q)
        \in \U.
$$
By the results of \cite{VKS_1996a} (see also \cite{DVKS_2001}), if $\theta_{K,M} \|\Delta_{K,M}\|_2^2  < m$, then $\bfA(K,M,\cdot): [0,\theta_{K,M}) \to \R_+$ is a strictly increasing convex function, with
$$
        \bfA(K,M,0) = 0,
        \qquad
        \left.
        \partial_q \bfA(K,M,q)
        \right|_{q = 0} = 0,
        \qquad
        \left.
        \partial_q^2 \bfA(K,M,q)
        \right|_{q = 0} > 0.
$$
From (\ref{bfQSL0}), it follows that, in the case $a = 0$,  the equations of
Theorem~\ref{main} lead to the the steady-state  gain matrices $K_0, M_0$ of the
standard Kalman filter.

\end{document}